\documentclass[a4paper,12pt]{article}

\usepackage[english]{babel}
\usepackage{authblk}
\usepackage{dsfont}
\usepackage{amsfonts}
\usepackage{mathrsfs}
\usepackage{amssymb}        
\usepackage{amsmath}
\usepackage{graphicx,caption}
\usepackage{float}
\usepackage[a4paper]{geometry} 
\usepackage{verbatim}
\usepackage{pstricks}
\usepackage{amsthm}
\usepackage{mathrsfs}
\usepackage{float}

\newtheorem{theorem}{Theorem}

\newtheorem{lemma}{Lemma}

\newtheorem{remark}{Remark}

\newcommand{\R}{\mathbb{R}}

\newcommand{\N}{\mathbb{N}}

\newcommand{\caL}{\mathcal{L}}

\newcommand{\dd}{\mathrm{d}}
\newcommand{\id}{\mathrm{i}}
\newcommand{\ed}{\mathrm{e}}

\newcommand{\Const}{\mathrm{C}}

\newcommand{\econd}{E}

\begin{document}

% title
\title{Subdiffusion in one-dimensional Hamiltonian chains with sparse interactions}

\author[1]{Wojciech De Roeck}
\affil[1]{
Instituut Theoretische Fysica, 
KU Leuven, 
3001 Leuven, Belgium
}

\author[2]{Francois Huveneers}
\affil[2]{
Ceremade,
UMR-CNRS 7534, 
Universit\'e Paris Dauphine, 
PSL Research University, 
75775 Paris cedex 16, 
France
}

\author[2]{Stefano Olla}

\date{\today}
\maketitle 

\begin{abstract}
We establish rigorously that transport is {slower than diffusive} for a class of disordered one-dimensional Hamiltonian chains. 
This is done by deriving quantitative bounds on the variance in equilibrium of the energy or particle current, as a function of time.
The slow transport stems from the presence of rare insulating regions (Griffiths regions). 
In many-body disordered quantum chains, they correspond to regions of anomalously high disorder, where the system is in a localized phase.
In contrast, we deal with quantum and classical disordered chains where the interactions, 
usually referred to as anharmonic couplings in classical systems, are sparse.
The system hosts thus rare regions with no interactions and, 
since the chain is Anderson localized in the absence of interactions, the non-interacting rare regions are insulating.    
Part of the mathematical interest of our model is that it is one of the few non-integrable models where the diffusion constant can be rigorously proven not to be infinite.
\end{abstract}

\begin{center}
\emph{Dedicated to Joel L. Lebowitz, for being a constant source of inspiration}
\end{center}

\section{Introduction}\label{sec: introduction}
One-dimensional Hamiltonian systems, also known as chains, sometimes exhibit anomalous transport properties, 
i.e.\ non-diffusive transport of locally conserved quantities like energy, particle number, momentum, etc.  
Whereas the possibility of superdiffusion is well-documented \cite{lepri_2016}, subdiffusion has not been studied as intensively.

In disordered systems, \emph{Anderson localization} can actually suppress transport entirely \cite{anderson_1958,gol1977pure,kunz_souillard_1980,froehlich_spencer_1983}. 
It can be realized e.g.\@ in disordered chains of  free fermions or, equivalently, in some quantum spin chains. 
Perhaps surprisingly, it can also be realized in classical disordered harmonic chains: 
The dynamics is linear and all the modes of the chain become exponentially localized in the presence of disorder.\footnote{
The first instances of disordered harmonic chains have been introduced in \cite{rubin_greer_1971,casher_lebowitz_1971}.
Because of momentum conservation, the localization length diverges in the bottom of the spectrum for these models, yielding a more complicated phenomenology. 
We will not deal with such cases here.}
For Anderson insulators, the current of conserved quantities is exponentially suppressed as a function of the length of the chain and, 
in equilibrium, the variance of the current of a conserved quantity, $C(t)$, does not grow with time, i.e.\@ $C(t) = \mathcal O (1)$ as $t\to \infty$ \cite{bernardin_huveneers_2013}.
See \eqref{eq: conductivity} below for a precise definition of $C(t)$.
For the quantum systems quoted above, there is strong evidence that, upon turning on interactions, 
localization persists at strong disorder and weak interactions, resulting in the so called many-body localized (MBL) phase
\cite{fleishman_anderson_1980,gornyi_mirlin_polyakov_2005,basko_aleiner_altshuler_2006,oganesyan_huse_2007,serbyn_papic_abanin_2013,imbrie_jsp_2016}, 
see also \cite{adp_2017} for reviews and more recent progress. 
There as well we expect $C(t) = \mathcal O (1)$ as $t\to \infty$.
%though the behavior of the AC conductivity at low frequency is sensitive to the presence of interactions \cite{gopalakrishnan_et_al_2015}.
In contrast, for classical systems, arbitrarily small anharmonic interactions are expected to destroy localization
\cite{dhar_lebowitz_2008,oganesyan_pal_huse_2009,mulanski_ahnert_pikovsky_shepelyansky_2009,basko_2011}.

Subdiffusion refers to the case where conserved quantities do evolve at the macroscopic scale, i.e.\@ $C(t) \to \infty$, 
but slow enough {so that  $C(t)/t \to 0$ as $t\to \infty$.}
Contrary to localization, subdiffusion does not challenge the basic principles of thermodynamics,
and we may expect the phenomenon to be more widespread. 
Interestingly, the existence of a subdiffusive ergodic phase has been predicted theoretically 
\cite{agarwal_gopalakrishnan_knap_mueller_demler_2015,gopalakrishnan_agarwal_demler_huse_knap_2016,agarwal_et_al_2017,potter2015universal, altman2015universal,lev2015absence} 
to occur in quantum chains with moderate disorder, close below the MBL transition point, see Figure~\ref{fig: phase diagram}.
This prediction has been verified numerically, although, to our best knowledge, 
it is not conclusively known whether the observed subdiffusion can unambiguously be ascribed to the theoretical predictions
\cite{luitz_et_al_2016,luitz_lev_2016,luitz_lev_2017,vznidarivc2016diffusive,kozarzewski2018spin,roy2018anomalous,mendoza2019asymmetry,schulz2018energy}. 
{This phenomenon is purely one-dimensional and rests on the presence of bottlenecks in the chain\footnote{More precisely, it rests on the sparsity of loops and the presence of bottlenecks, see \cite{de2019sub}.}}. 
For generic classical systems, one does not expect the existence of an analogous subdiffusive phase, see \cite{basko_2011} for a detailed analysis.
 
In this paper, we present a mathematically rigorous derivation of slower-than-diffusive transport in a class of quantum and classical models, 
suggested to us by D.~A.~Huse. 
In contrast to most models studied in the literature, the interactions considered here are \emph{sparse} (but not weak),  
turning them on at each site independently with probability $p<1$.
A similar set-up was recently studied in \cite{nachtergaele_reschke_2019} and mathematical results on long transmission times were derived. 
As it turns out, the existence of a subdiffusive phase can then be shown to hold both for classical and quantum Hamiltonians 
(even though the corresponding classical systems are still not expected to have a MBL phase). 

The rest of this paper is organized as follows. 
Below, we provide a more detailed heuristic explanation of the phenomenology. 
In Section \ref{sec: models and results}, we define properly the classcial and quantum Hamiltonians studied in this paper, 
and we state precise conditions for the existence of a subdiffusive phase in Theorem \ref{th: disordered chain} for the classical chain, 
and in Theorem \ref{th: disordered fermion chain} for the quantum chain. 
The proof of Theorem \ref{th: disordered chain} is presented in Section \ref{sec: proof disordered} 
and the proof of Theorem \ref{th: disordered fermion chain} in Section~\ref{sec: proof disordered fermion}.

\subsection{Griffiths regions in disordered chains}
Following \cite{agarwal_gopalakrishnan_knap_mueller_demler_2015}, let us explain heuristically the origin of subdiffusion in disordered quantum systems. 
For the sake of concreteness, let us consider the celebrated ``disordered XXZ spin chain":
\begin{equation}\label{eq: true xxz}
	H \; = \;  \sum_{x=1}^{L-1} \big( J (S^{(1)}_x S^{(1)}_{x+1}  +   S^{(2)}_x S^{(2)}_{x+1} ) + g S^{(3)}_{x} S^{(3)}_{x+1} \big) + \sum_{x=1}^L \omega_x S^{(3)}_{x} 
\end{equation}
which has become the standard model for MBL. Here, $S_x^{1,2,3}$ are the spin-$\tfrac12$ spin operators (Pauli matrices) acting on site $x$ of a chain (see e.g.\ \cite{alicki_fannes_2001} for more background and a full definition) and $J,g$ are parameters of the model. 
The on-site fields $\omega_x$ are assumed to be independent, identically distributed (iid) random variables, 
e.g.\ distributed according to a normal distribution with zero mean and standard deviation $W$, so that the parameter $W$ plays the role of the ``disorder strength".
The specific choice of the distribution plays no crucial role, and in particular the fact that $\omega_x$ is unbounded turns out to be eventually irrelevant.
For $g=0$, the system is an Anderson insulator, since it can be mapped to free fermions by a Jordan-Wigner transformation, see e.g.~\cite{abdul_et_al_2017}.
Hence, it is localized at all values of $W>0$.
Instead $g>0$ brings in interactions among the spins and the system is expected to stay localized only at large disorder $W>W_c$,
while it becomes delocalized for smaller values of $W$, see Figure~\ref{fig: phase diagram}.

For $W > W_c$, i.e.\@ in the localized phase, the resistance of of a stretch of length $\ell$ is given by 
\begin{equation}\label{eq: resistance localized}
	r\sim \ed^{2\ell/\xi} \, 
\end{equation}
where $\xi=\xi(W)$ is the \emph{localization length} characterizing the localized system.
Strictly speaking, the resistance $r$ depends on the disorder realization, 
but for large $\ell$, we can neglect the fluctuations around the average value in \eqref{eq: resistance localized}.
Here, we focus on a system where the disorder strength is actually too small to bring about localization, i.e.\@ $W <W_c$,
but local anomalously large disorder values can be strong enough to put small stretches of the systems in the localized phase,  
i.e.\@ these regions would be in the localized phase if they would be disconnected from the rest of the chain.
For example, it can be that on a stretch of size $\ell$, the value of the empirical standard deviation $\tilde W$ is larger than $W_c$.
Such a region is called a \emph{Griffiths region}.
As can be inferred from the previous example, a Griffiths region of length $\ell$ occurs in a given place with probability $\ed^{-c(W)\ell}$, 
where the large deviation rate $c(W)$ is expected to vanish as $W$ approaches the critical value $W_c$.  
If we assume for simplicity that all such Griffiths regions have the same localization length $\xi_*$, 
then the probability distribution of the resistance $r$ of a single Griffiths region is given roughly by
$$
	p(r)\dd r  \;  \propto   \;  \ed^{-c(W)\ell} \ed^{-2\ell/\xi_*} \dd r   \; \propto  \;  r^{-(1 + \xi_* c(W)/2)}   \dd r \, .
$$ 
For $\xi_*c(W)<2$, we notice that ${E}(r)=\infty$. 

To understand how Griffiths regions can lead to a subdiffusive behavior when $\xi_*c(W)<2$, 
we represent the chain as a Ohmic circuit of resistances $(r_i)_{1 \le i \le L}$ in series, where the resistances are the Griffiths regions
(sites outside of any Griffiths region may simply be assumed to have a resistance of order 1). 
The resistances $(r_i)_{1 \le i \le L}$ are i.i.d.\@ in very good approximation. 
The total resistance is $R=\sum_{i=1}^L r_i$ and the conductivity, which is more suited for our discussion, is given by
$K = L/R $, where $L$ stands for the length of the chain. 
If the resistances $r_i$ have very heavy tails, in particular if ${E}(r_i) = \infty$, then the conductivity vanishes almost surely: $K \to 0$, as $L \to \infty$.
Instead, if ${E}(r_i) < \infty$, then $K>0$.
The system is thus subdiffusive if $\xi_*c(W)<2$, i.e.\@ if one is not too far from the localized phase, see Figure~\ref{fig: phase diagram}.

\begin{figure}[H] 
	\begin{center}
    \vspace{5mm}
  	\includegraphics[width=10cm]{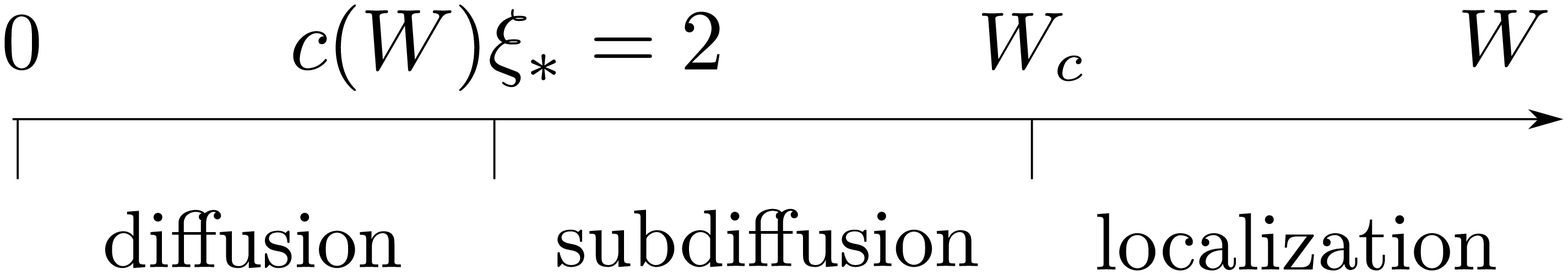}
    \vspace{3mm}
    \end{center}
    \caption{Transport in a generic quantum chain as a function of disorder strength $W$.} 
    \label{fig: phase diagram}
 \end{figure}

\subsection{Systems without genuine localized phase}\label{sec: systems without genuine}
As said, for classical systems with generic anharmonic interactions, one does not expect a genuinely localized phase. 
This is related to the fact that, at positive temperature in the thermodynamic limit, all typical configurations are expected to contain some chaotic spots
that will eventually break the quasi-periodic behavior of the dynamics.
Yet, the dynamics in certain classical systems can certainly be extremely slow and glassy, see \cite{de_roeck_huveneers_2019} for a recent review.  
The biggest difference with quantum chains is that the presence of an insulating Griffiths region is mainly singled out by an atypical configuration of the oscillators.
Indeed, large disorder fluctuations will only ensure that a larger proportion of the configurations is insulating, never all of them in generic cases.
Therefore, if the dynamics is eventually ergodic as one expects, a specific Griffiths region exists only for a finite time. 
A detailed quantitative analysis done in \cite{basko_2011} shows that the lifetime of Griffiths regions, compared to their resistance, 
is too short to bring about subdiffusion. 
While this conclusion pertains to generic anharmonic disorder chains of oscillators, 
we will see that it does not apply to the more specific class of classical Hamiltonians considered in this paper.

\subsection{Sparse interactions}
The models treated in this paper have the following specific property:
\emph{Interactions} or \emph{anharmonic couplings} are present at each site with probability $p$ and absent with probability $1-p$. {In a typical realization, there are hence arbitrarily long stretches} of sites with no interaction and these stretches play the role of Griffiths regions. 
These Griffiths regions are thus simply Anderson insulators. 
In particular, the issue raised above for classical chains is avoided: 
All the modes of the chain are exponentially localized and the region behaves as a perfect insulator 
(i.e.\@ with a conductivity vanishing exponentially with the length) independently of the region in phase space that the system occupies. 

Finally, one may wonder whether our quantum model is actually localized in the regime where we predict subdiffusion. 
The rigorous treatment of this question goes way beyond the present paper, but we can speculate on the phase diagram of our model.
Let us restrict to $p \ll 1$, i.e.\ very sparse interactions. 
In the regions where the interactions are switched off, our model is simply an Anderson insulator and hence localized regardless of the disorder strength. 
Then, the question is whether the sparse regions with interaction can delocalize such a system. 
This question has been investigated and answered in previous work, \cite{de_roeck_huveneers_2017, luitz_huveneers_de_roeck_2017}, 
and the upshot is that we should expect the system to be localized whenever the localization length of the non-interacting system is small enough, and delocalized otherwise.   
Therefore, our model is indeed expected to have a localization-delocalization transition and the region where we prove subdiffusion overlaps with presumed delocalized phase.

 \begin{figure}[H] 
        \begin{center}
     \vspace{5mm}
            \includegraphics[width=10cm]{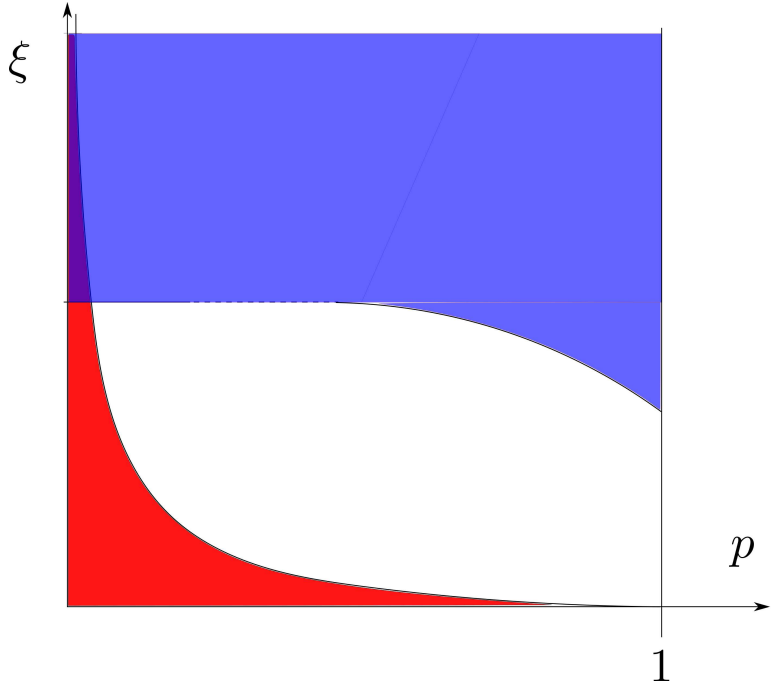}
             \vspace{3mm}
        \end{center}
        \caption{
        Our proof of slower than subdiffusive transport applies in the red area of the $p-\xi$ plane.  
        For our quantum model, the blue area indicates where we expect a delocalized (non-MBL) phase, and we see that these regions have a non-zero intersection} \label{fig: sparse diagram}
 \end{figure}

\section{Models and Results}\label{sec: models and results}
We study in details two different Hamiltonian chains. Each of them is representative of a larger class of systems. 
Let  $\N_L=\{1,\ldots,L\}$ be the set of sites, and we will eventually take the length $L\to \infty$.

\subsection{Anharmonic disordered chain}\label{sec: classical chain}
Let us first consider the disordered classical anharmonic chain, described by the following Hamiltonian on $\R^{2L}$: 
\begin{equation}\label{eq: disordered chain}
	H (q,p)
	\; = \; 
	\sum_{x=1}^L h (q_x,q_{x+1},p_x,\omega_x,\tau_x)
	\; = \; 
	\sum_{x=1}^L \left( \frac{p_x^2}{2} + \omega_x^2 \frac{ q_x^2}{2} + g \tau_x \frac{q_x^4}{4} + g_0  \frac{(q_{x+1} - q_x)^2}{2} \right)
\end{equation}
with Neumann boundary conditions at both ends, i.e.\@ $q_{L+1} = q_L$. 
We take the anharmonic coupling $g$ and the harmonic coupling $g_0$ to be non-negative.

The frequencies of the uncoupled oscillators, $(\omega_x)_{x\in \N_L}$, 
are independent and identically distributed (iid) random variables with a bounded density with compact support not containing $0$, 
i.e.\@ there exists $\omega_- > 0$ such that 
$$
	\omega_x^2 \; \ge \; \omega_-^2 \; > \; 0 \qquad \text{almost surely}.
$$
The latter condition is most likely not essential but guarantees that the Hamiltonian $H$ is strictly convex, 
a feature that turns out to be convenient to establish the decay of correlations in the Gibbs ensemble. 
The variables $\tau_x$ are iid random variables, independent of $(\omega_x)$,
and with value in $\{ 0 , 1 \}$. 
We define 
$$
	p \; = \; P (\tau_x = 1) .
$$

The dynamics is given by Hamilton's equations: 
\begin{equation}\label{eq: Hamilton}
	\dot q \; = \; p, \qquad \dot p = - \nabla_q H(q,p)
\end{equation}
and the generator of the dynamics is given by the usual Liouville operator: 
$$
	\mathcal L \; = \; p \cdot \nabla_q - \nabla_q H \cdot \nabla_p  \, .
$$
The energy is the only obvious locally conserved quantity. 
For simplicity let us write $h_x = h (q_x,q_{x+1},p_x,\omega_x,\tau_x)$.
For $ x \in \N_L$, we define the energy current $j_x$ between the site $x$ and $x+1$ as
$$
	j_x \; = \;  \{h_{x+1},h_{x}\}   \; = \;  - g_0 p_{x+1} (q_{x+1} - q_x) 
$$
with the conventions $q_0 = q_1$ and $q_{L+1} = q_L$, so that $j_0 = j_L = 0$ as imposed by the Neumann boundary conditions in \eqref{eq: disordered chain}.
This definition ensures that the currents satisfy the continuity equation
$$
	\mathcal L h_x \; = \; - \nabla j_x \; = \; j_{x-1} - j_x \, .
$$

\subsubsection{The harmonic chain: $g = 0$}
At $g=0$, the oscillator chain is an Anderson insulator due to disorder.  
Indeed, the equations of motion \eqref{eq: Hamilton} are linear and the equation for $p$ takes the form 
\begin{equation}\label{eq: anderson hamiltonian}
	\dot p_x \;=\; -\sum_y \mathcal H_{x,y}  q_y
	\qquad \text{with}\qquad
	\mathcal H \; = \; V -g_0 \Delta
\end{equation}
where $V_{x,y}=\omega_x^2 \delta_{x,y}$ and where $\Delta$ is the discrete laplacian on $\N_L:= \{ 1, \dots , L\}$ with Neumann boundary conditions: 
$$
	\Delta f(x) 
	\; = \;
	f(x-1) - 2f(x) + f(x+1) 
	\qquad \text{for} \qquad x \in \N_L \, 
$$
with the convention $f(0)=f(1)$ and $f(L+1) = f(L)$.
The random operator $\mathcal H$ on $l^2(\N_L ,\mathbb R)$ is known as  the \emph{Anderson Hamiltonian} for a single quantum particle \cite{anderson_1958}.  
It is a celebrated result \cite{gol1977pure,kunz_souillard_1980} that its eigenvectors are localized, see \cite{cycon_et_al_1987,damanik_2011,ducatez_2017} for pedagogical expositions. 
\begin{lemma}[Anderson localization, Theorem IX.1 in \cite{kunz_souillard_1980}]\label{lem: anderson localization}
Let $(\psi_k)_{1 \le k \le L}$ be an orthonormal basis of eigenvectors of $\mathcal H$ (depending hence on the realization of $(\omega_x)$). There are $\Const < + \infty$ and $\xi > 0$, such that, for any $L$ and $1 \le x, y \le L$, 
\begin{equation}\label{eq: kunz souillard bound}
	{E} \left( \sum_{k=1}^L |\psi_k(x) \psi_k(y)| \right)
	\; \le \; 
	\Const \ed^{- |x-y|/\xi} 
\end{equation}
\end{lemma}

\noindent
(since we assume that the distribution of $\omega_x$ is compactly supported, we may take $A=\R$ in Theorem IX.1 in \cite{kunz_souillard_1980}).
We refer to $\xi$ as the localization length of the system, though it is more precisely an upper bound for the localization length uniform in the energy. 
We note that the localization length $\xi$ depends thus only on $g_0$ and on the distribution of $(\omega_x)$.

\subsubsection{The anharmonic chain: $g \ge 0$} 
We finally denote by $\langle \cdot \rangle_\beta$ the expectation with respect with the Gibbs measure at inverse temperature $\beta > 0$, 
i.e.\@ the probability measure on $\R^{2L}$ with density 
\begin{equation}\label{eq: gibbs state disorder chain}
	\rho_\beta(q,p) \; = \; \frac{\ed^{- \beta H(q,p)}}{Z(\beta)}
\end{equation}
where $Z(\beta)$ ensures the normalization (the partition function).    
Note that $\langle j_x \rangle_\beta=0$ for any $x$, because the currents $j_x$ are odd in the $p$-variables.
The thermal conductivity $\kappa$ at inverse temperature $\beta$ is defined by 
\begin{equation}\label{eq: conductivity}
	\kappa
	\; = \; 
	\beta^2 \lim_{t\to\infty} \frac{C(t)}{t}
	\qquad \text{with} \qquad  
	C(t) \; = \;  \limsup_{L\to\infty}
	\bigg\langle \bigg(\int_0^t \dd s \frac{1}{\sqrt{L}}\sum_{x=1}^{L-1} j_x (s) \bigg) ^2\bigg\rangle_\beta \, .
\end{equation}
The expression for $C(t)$ is well-defined because we view $(\omega_x,\tau_x)_{x \in \N_L}$ as restrictions of a single infinite array $(\omega_x,\tau_x)_{x \in\mathbb{N}}$ of i.i.d.\@ random variables, so that $C(t)$ is a limit superior of random variables defined on a single probability space. 
Furthermore, in the definition of $C(t)$, we expect that the limit superior is actually a limit, almost surely, but we are not able to prove this.  
This is however irrelevant for establishing that transport is slower than diffusive, which is equivalent to the fact that {$C(t)/t\to 0$} almost surely, as $t \to \infty$. 
In Theorem~\ref{th: disordered chain} below, we provide an upper bound on $C(t)$
implying slower than diffusive behavior of the transport for sufficiently small $p$. {We recall that $C(t)=\mathcal{O}(f(t)),t \to \infty$ means that $\frac{C(t)}{f(t)}$  is bounded as $t\to\infty$. }

\begin{theorem}\label{th: disordered chain}
Fix any $\beta > 0, g \ge 0, g_0>0$.  If the inequality
$$\gamma \; := \; \frac{4}{1+(3 \xi  \log (\frac1{1-p}) )^{-1}}   \; <\; 1,   $$
holds (i.e.\@ if $p < 1 - \ed^{-1/9\xi}$), with $\xi$ the localization length given in \eqref{eq: kunz souillard bound}, 
then %the dynamics is subdiffusive, more precisely
$$
	C(t) = \mathcal O( (\log t)^5 t^{ \gamma } ),\qquad t \to \infty,\qquad \text{almost surely}. 
$$
\end{theorem}

\begin{remark}
For $p=0$ (equivalently $g=0$) our bound gives $\gamma=0$. However, in that case the dynamics is localized and it is known that $C(t) = \mathcal O (1)$ as $t \to \infty$, see the remarks below Theorem~1 in~\cite{bernardin_huveneers_2013}. 
\end{remark}

\begin{remark}
By direct adaptation of the proof, the on-site potential $q \mapsto g q_x^4$ may be replaced by a more general potential $q \mapsto V(\theta_x q)$, 
where $V$ is a smooth, convex, local function, growing at most polynomially at infinity as well as all its derivatives,
and where $\theta_x q$ denotes the shifted configuration by $x$.
\end{remark}

\subsection{Disordered Fermion chain}  \label{sec: fermion chain}

This model differs from the one discussed above because it is quantum and, more importantly, 
it is believed to have a genuinely localized interacting phase (MBL) for some range of parameters. 

To define the fermion chain, we need the fermionic creation and annihilation operators  $c^{\dagger}_x,c_x$ that satisfy the \emph{canonical anticommutation relations}
$$
	[  c^{\phantom{\dagger}}_x, c^{\dagger}_{x'} ]_+ =\delta_{x,x'},\qquad  [  c^{\phantom{\dagger}}_x, c^{\phantom{\dagger}}_{x'} ]_+= [  c^{\dagger}_x, c^{\dagger}_{x'} ]_+=0
$$
where $[ a,b]_+=ab+ba$ is the anticommutator and $x,x' \in \N_L$.   We refer again to standard texts, e.g.\ \cite{alicki_fannes_2001}, for an explicit representation of $c^{\phantom{\dagger}}_x,c_x^\dagger$ as operators on the appropriate finite-dimensional Hilbert space $\mathcal{F}_L$.  
We also introduce the fermionic number operators $n_x=c^{\dagger}_x c^{\phantom{\dagger}}_{x}$ counting the number of fermions ($0$ or $1$) at site $x$.  
The Hamiltonian models interacting spinless fermions: 
\begin{equation}\label{eq: fermionic chain}
	H \; = \; \sum_{x=1}^{L-1}    \big( J (c^\dagger_x c^{\phantom{\dagger}}_{x+1}+  c^\dagger_{x+1} c^{\phantom{\dagger}}_x)     +  g\tau_x n_x n_{x+1} \big) + \sum_{x=1}^L\omega_x n_x 
\end{equation}
where the random variables $\omega_x,\tau_x$ have the same law as those in Section \ref{sec: classical chain}.

This model can be mapped by the Jordan-Wigner transformation to a disordered variant of the XXZ-model:
\begin{equation}\label{eq: xxz}
	H 
	\; = \; 
	\sum_{x=1}^{L-1}   \big( J (S^{(1)}_x S^{(1)}_{x+1}  +   S^{(2)}_x S^{(2)}_{x+1} ) +  g\tau_x S^{(3)}_{x} S^{(3)}_{x+1} \big) 
	+ \sum_{x=1}^L \omega_x S^{(3)}_{x}.
\end{equation}
In particular, setting $\tau_x=1$ yields back the disordered XXZ chain introduced in \eqref{eq: true xxz}.
For reasons of convenience, we work however with the fermionic chain. 

\subsubsection{Dynamics} \label{sec: dynamics}
In general, by observables, we mean any Hermitian  operator $a$ acting on $\mathcal{F}_L$. 
The dynamics of observables is given by the Heisenberg evolution equation 
$$
	\frac{\dd}{\dd t} a_t \; = \; \mathcal L a_t \; := \;  \id [H,a]  ,   \qquad  a_0=a
$$ 
The model, has two conserved quantities, the energy, i.e.\ the Hamiltonian itself and particle number $N=\sum_x n_x$. Indeed, one checks easily that $[H,N]=0$.  For simplicity, we discuss only transport of the particle number. 
The natural particle current operator $j_x$
is defined by 
$$
j_x= \id [  J (c^\dagger_x c_{x+1}+  c^\dagger_{x+1} c_x),n_x]= \id J (c^\dagger_{x+1} c_x-c^\dagger_x c_{x+1})
$$ 
and it satisfies the continuity equation
$$
	\mathcal L n_x \; = \; - \nabla j_x \; = \; j_{x-1} - j_x \, .
$$

As for the equilibrium state of the system, we choose the  density matrix $\rho_{\mu} \propto  \ed^{-\mu N} $.  The expectation value of observables in this state is then given by
\begin{equation}\label{def: quantum state}
\langle a \rangle_\mu = \frac{\mathrm{Tr} (\rho_{\mu} a)}{\mathrm{Tr} (\rho_{\mu})} 
\end{equation}
where  $\mathrm{Tr}(\cdot)$ is the trace over the  Hilbert space $\mathcal{F}_L$.
Just as in Section \ref{sec: classical chain}, there are no subtle mathematical issues here because our setup is at finite volume. 
Note finally that $\langle j_x \rangle_\mu=0$. This follows for example by considering the gauge transformation $c_x \to \ed^{\id \theta_x} c_x, c^\dagger_x \to \ed^{-\id \theta_x} c^\dagger_x $, under which the state $\langle \cdot\rangle_\mu$ is invariant.  

\subsubsection{The non-interacting chain}
When $g = 0$, the Hamiltonian reduces to 
\begin{equation}\label{eq: hb as free}
H=\sum_{ x=1 }^{L-1}    J (c^{\dagger}_{x+1 c_x }+   c^{\dagger}_x c_{x+1})    +  \sum_{x=1}^L\omega_x n_x 
= \sum_{x,x' \in \N_L}  c^\dagger_{x'} \widetilde{ \mathcal{H}}_{x,x'}   c_x
\end{equation}
where 
$$ \widetilde{ \mathcal{H}} = J\Delta +2J   + V,
$$ 
with $V_{x,y}=\omega_x\delta_{x,y}$, is again the Anderson Hamiltonian. It only differs from $\cal H$ of Section \ref{sec: classical chain} by irrelevant constants and a different distribution of on-site disorder.  Obviously, Lemma \ref{lem: anderson localization}, formulated for $\mathcal{H}$, applies to $\widetilde{ \mathcal{H}}$ just as well and it yields an upper bound $\xi$ for the localization length that depends only on $J$ and on the distribution of $\omega_x$.
The equality in \eqref{eq: hb as free} means that the Hamiltonian $H$ is the \emph{second quantization} of $\widetilde{\mathcal H}$, or, in more intuitive terms, that $H$ describes no-interacting fermions.

Just like in the classical anharmonic chain, we consider the current-current correlation
\begin{equation}\label{eq: fermion conductivity}
	C(t) \; = \;  \limsup_{L\to\infty}
	\bigg\langle \bigg(\int_0^t \dd s \frac{1}{\sqrt{L}}\sum_{x=1}^{L-1} j_x (s) \bigg) ^2\bigg\rangle_\mu \, .
\end{equation}
The slower than diffusive behavior of the dynamics is equivalent to the fact that $C(t)/t \to 0$ as $t \to \infty$.   In Theorem~\ref{th: disordered fermion chain} below, we provide an upper bound on $C(t)$ guaranteeing slower than diffusive dynamics for sufficiently small $p$.

\begin{theorem}\label{th: disordered fermion chain}
Fix any $\mu > 0, g \ge 0, J>0$.  If the inequality
$$\gamma \; := \; \frac{4}{1+(3 \xi  \log (\frac1{1-p}) )^{-1}}   \; <\; 1,   $$
holds (i.e.\@ if $p < 1 - \ed^{-1/9\xi}$), with $\xi$ the localization length given in \eqref{eq: kunz souillard bound}, then the dynamics is slower than diffusive: 
$$C(t) = \mathcal O((\log t)^5t^{ \gamma } ),\qquad t \to \infty,\qquad \text{almost surely}. $$ 
\end{theorem}

\section{Proof of Theorem~\ref{th: disordered chain}}\label{sec: proof disordered}

For $p=0$, the chain is harmonic and it is known that $C(t) = \mathcal O(1)$ as $t\to \infty$ (see the remarks below Theorem~1 in~\cite{bernardin_huveneers_2013}). 
From now on, we assume $p >0$, hence the model is non-integrable.   Throughout the proof, constants $\Const$ can always depend on $\beta,g,g_0$ and the distribution of $\omega_x$ (unless stated otherwise),  but not on the length $L$. 

\subsection{Decay of static correlations}

We need decay of spatial correlations in the Gibbs state $\langle \cdot \rangle_\beta$. 
Given a function $u$ on $\R^{2L}$, we denote by $\mathrm{supp}(u)\subset \N_L$ the smallest set of points such that 
$u$ is constant as a function of $q_x, p_x$ if $x \notin \mathrm{supp}(u)$ (it is thus not the usual support of a function in $\R^{2L}$).
Given two functions $u,v$, we denote by $d(u,v)\ge 0$ the distance between $\mathrm{supp}(u)$ and $\mathrm{supp}(v)$. 
If $u$ denotes a smooth function on $\R^{2L}$, we write 
$$
	\nabla u = (\nabla_q u , \nabla_p u) = \big( (\partial_{q_x} u )_{1 \le x \le L} , (\partial_{p_x} u )_{1 \le x \le L} \big) . 
$$
The Lemma below yields the needed decay of correlations; it is a special case of more general results stated in \cite{helffer_1998,helffer_1999}, 
see also \cite{bodineau_helffer_2000} for a review.
 
\begin{lemma}[decay of correlations]\label{lem: decay of correlations}
1. Given $r\in \N$, there exists $\Const_r< + \infty$ such that, for any $1\le x \le L$, 
$$
	|\langle q_x^r\rangle_\beta |\le \Const_r, \qquad |\langle p_x^r\rangle_\beta |\le \Const_r.
$$
2. There exist $\Const,\zeta<+\infty$ such that, given two polynomials $u,v$ on $\R^{2L}$ with $\langle u \rangle_\beta = \langle v \rangle_\beta = 0$,
$$
	|\langle u v \rangle_\beta|
	\; \le \; 
	\Const \ed^{- d(u,v)/\zeta} \langle\nabla u \cdot \nabla u\rangle_\beta^{1/2}\langle\nabla v \cdot \nabla v\rangle_\beta^{1/2}
$$
These bounds hold almost surely and the constants  $\Const, \Const_r, \zeta$ are independent of $g$. 
\end{lemma} 

\begin{proof}
Polynomials are obviously integrable with respect to the measure $\mu_\beta$ at fixed length $L$. 
Let us start with Part 2, following Section 4 in \cite{helffer_1998}.
We first notice that there is no need for the restriction $u=q_i$ or $u=p_i$ and $v=q_j$ or $v=p_j$ for some $1 \le i,j \le L$ in \cite{helffer_1998} 
(see also the statement of Theorem 3.1 in \cite{bodineau_helffer_2000}).
Next, the function
$$
	(q,p) \; \mapsto \;
	\beta \sum_{x=1}^L \left( \frac{p_x^2}{2} + \omega_x^2 \frac{ q_x^2}{2} + g \tau_x \frac{q_x^4}{4} + g_0  \frac{(q_{x+1} - q_x)^2}{2} \right)
$$ 
is strictly convex with a Hessian bounded below by $\beta \min (1, \omega_-^2)$.
From there, the result follows as in \cite{helffer_1998}, if $\zeta$ is taken large enough compared to $g_0$ and $g_0/\omega_-^2$. 
Let us next consider Part 1. 
We notice that $\langle p_x^r\rangle_\beta = (\beta/2 \pi)^{1/2} \int_\R  y^r \ed^{- \beta y^2 /2} \dd y$, 
so that we only need to deal with the bound on $\langle q_x^r\rangle_\beta$.
It is enough to consider $r$ even since $\langle q_x^r\rangle_\beta = 0$ otherwise. 
For $r=2$, the result follows from part 2 of the Lemma taking $u=v=q_x$. 
For $r>2$, we obtain similarly $\langle q_x^r \rangle_\beta \le \Const \langle q_x^{r-2}\rangle_\beta$, and the result is obtained recursively. 
\end{proof}

\subsection{Splitting of the harmonic Hamiltonian}\label{sec: splitting of harmonic}

We now exploit the fact that the system is an Anderson insulator in the regions where the anharmonic potential is absent, i.e.\@ where $\tau=0$. The upshot is Lemma \ref{lem: bounds on gamma}, whose significance is explained below. 
The material below is largely taken from  \cite{bernardin_huveneers_2013}, 
but we have tried to streamline the presentation to make the reasoning more transparant. 

Throughout this section, we fix a realization of $(\tau_x)_{x \ge 1}$. 
Given $\ell \in \N$, let
\begin{equation}\label{eq: G not set}
	G_0(\ell) 
	\; = \;
	\{ 
	x \in \N \, : \,  \tau_y = 0 \text{ for all $y \in \N $ s.t.\  } |y-x|\le \ell 
	\} \, .
\end{equation} 
We also fix an element $x\in G_0(\ell)$ and we denote $B=\{x-\ell,\ldots, x+\ell\}$, such that $\tau_y=0$ for $y \in B$.
Until the end of this section \ref{sec: splitting of harmonic}, the expectation $\econd$ is assumed to be conditioned on $\tau_y=0, y \in B$ and we do not repeat this. 
\\
 
In general, let $H_X$ be the Hamiltonian restricted to a finite set $X \subset \N$, i.e.\ retaining only terms in \eqref{eq: disordered chain} whose support is in $X$. Similarly,  we define the restriction of
the Anderson operator  $\mathcal{H}_X= \sum_{x,y \in X} \mathcal{H}_{x,y}$ and of the
Liouville operator $\cal L_X=\{H_X,\cdot\}$.

\subsubsection{A priori left-right splitting}
We first consider an obvious splitting $H_B= H_{\mathrm{L}} + H_{\mathrm{R}}$ into a left ($\mathrm{L}$) and right ($\mathrm{R}$) part with respect to the midpoint $x$:
$$
 H_\mathrm{L} \; = \;H_{\{y\in B: y<x\}},\qquad  H_{\mathrm{R}}   \; = \;	H_B - H_{\mathrm{L}}.
$$
Note that hence $H_{\mathrm{R}} = H_{\{y\in B: y\geq x\}} + g_0  \tfrac12 {(q_{x} - q_{x-1})^2}$
The functions $H_{\mathrm{L}},H_{\mathrm{R}}$ do not commute with $H_B$ in general, reflecting the fact that energy can be transported. 

Our aim in this section is to find a modified left-right splitting $H_B=\widetilde{H}_\mathrm{L}+\widetilde{H}_\mathrm{R}$ which does satisfy the invariance property $\{H_B,\widetilde{H}_{\mathrm{L},\mathrm{R}}\}=0$, reflecting the spatial localization of energy.

\subsubsection{An invariant splitting}
 Let $(\psi_k)_{k \in \mathcal{I}}$ be an orthonormal basis of real eigenfunctions of $\mathcal{H}_B$, with $\mathcal{I}$ an index set, $|\mathcal{I}|=|B|$, and let $(\nu_k^2)$ be the corresponding eigenvalues, positive due to the positive-definitiveness of $\mathcal{H}$.  For any $k \in  \mathcal I$, the function
\begin{equation}\label{eq: local e}
	e_k \; = \; \frac12 \left( \langle p,\psi_k \rangle^2 + \nu_k^2 \langle q, \psi_k \rangle^2 \right)
\end{equation}
represents the energy of eigenmode $k$. 
The functions $e_k$ have two remarkable properties, following from standard considerations, 
\begin{enumerate}
\item   $
	\mathcal L_B e_k \; = \; 0 $   for all $ k \in \mathcal{I}$. 
\item   $H_B=\sum_k e_k$.  
\end{enumerate}
The splitting $H_B=\widetilde{H}_{\mathrm{L}} + \widetilde{H}_{\mathrm{R}} $ that we propose is defined by
\begin{equation}\label{eq: expression tilde}
\widetilde{H}_{\mathrm{L}} =\sum_{y}   \chi(y<x)\sum_{k} |\psi_k(y)|^2  e_k ,\qquad    \widetilde{H}_{\mathrm{R}} =\sum_{y}  \chi(y \geq x)\sum_{k}|\psi_k(y)|^2
 e_k 
 \end{equation}
 From the above properties, it is clear that this is indeed a splitting and that $\{H_B,\widetilde{H}_{\mathrm{L},\mathrm{\mathrm{R}}}\}=0 $, but it is not clear a-priori in which sense this splitting is similar to $H_B=H_{\mathrm{L}}+H_{\mathrm{R}} $ and we exhibit this now.
Note first that both $H_{\mathrm{L},\mathrm{R}},\widetilde{H}_{\mathrm{L},\mathrm{R}} $ are linear combinations of $p_zp_w, q_zq_w$ with $z,w \in B$. 
Let us call 
\begin{equation} \label{eq: difference splittings}
\widetilde{H}_{\mathrm{L}}- H_{\mathrm{L}} = \sum_{z,w \in B} \left(\gamma^{(\mathrm{L})}_{z,w} p_zp_w + \alpha^{(\mathrm{L})}_{z,w} q_z q_w \right), \qquad   \widetilde{H}_{\mathrm{R}}- H_{\mathrm{R}} = \sum_{z,w \in B} \left(\gamma^{(\mathrm{R})}_{z,w} p_zp_w + \alpha^{(\mathrm{R})}_{z,w} q_z q_w \right)
\end{equation}
with the coefficients $\gamma^{(\mathrm{L},\mathrm{R})}_{z,w},\alpha^{(\mathrm{L},\mathrm{R})}_{z,w}$ chosen symmetric in $z,w$. 
The functions \eqref{eq: difference splittings} are localized around $x$ in the sense that $|\gamma_{z,w}|$ typically decay exponentially in $|z-x|+|w-x|$. We will not state this in full generality because the following lemma suffices for our purposes.
\begin{lemma}\label{lem: bounds on gamma}

\begin{enumerate}
\item For any $z,w$, $|\gamma^{(\mathrm{L},\mathrm{R})}_{z,w}| <\Const $ and $|\alpha^{(\mathrm{L},\mathrm{R})}_{z,w}| <\Const$
\item  If one of $z,w$ is at the boundary of $B$, (i.e.\ at distance $1$ of $B^c$), then 
$$
\econd(|\gamma^{(\mathrm{L},\mathrm{R})}_{z,w}|+|\alpha^{(\mathrm{L},\mathrm{R})}_{z,w}|)   \leq \Const e^{-\ell/\xi}
$$
\end{enumerate}
\end{lemma}
\begin{proof}
We start with $(1)$. 
By combining \eqref{eq: local e} and \eqref{eq: expression tilde}, we get the explicit expression
 \begin{equation}
	\gamma^{\mathrm{R}}_{w,z} 
	\;  = \; 
	\frac{1}{2}\sum_{y \geq x} \left(\sum_{k} |\psi_k(y)|^2 \psi_k(w)\psi_k(z)       - \delta_{y,z} \delta_{z,w}\right)
	\label{eq: gamma1}
\end{equation}
(for $\mathrm{L}\to \mathrm{R}$, we replace $y \geq x$ by by $y < x$).
Using $\sum_y |\psi_k(y)|^2=1$ and $\sum_k |\psi_k(z)|^2=1$, we then estimate
$$
|\gamma^{\mathrm{R}}_{w,z} | \leq  \frac{1}{2}\sum_{k} |\psi_k(w)||\psi_k(z)| +\frac{1}{2}  \leq  1
$$
and similarly for $\mathrm{R}\to \mathrm{L}$. For $\alpha$ instead of $\gamma$, we bound first $\nu^2_k\leq \Const$ using that $\mathcal{H}_B$ is bounded by $\Const$ as an operator on $l^2(B)$, because the distribution of $(\omega_x)$ has bounded support.  Thereafter, the argument is analogous to that for $\gamma$. \\
Now to $(2)$. For the sake of concreteness, let us take $z$
 to lie at the left boundary, i.e. $z=x-\ell$
%$\gamma^{(\mathrm{R})}_{zw}$ with  $z$  and $w \neq z$.
 We start from  \eqref{eq: gamma1} and we use $  |\psi_k(w)| \leq 1$ to obtain
$$
\econd(|\gamma^{\mathrm{R}}_{w,x-\ell}|) \leq  \sum_{y \geq x}\econd \big(\sum_k  |\psi_k(y)| |\psi_k(x-\ell)|  \big) \leq  \Const e^{-\ell/\xi},\qquad  z=x-\ell
$$
where the crucial last inequality follows from Lemma \ref{lem: anderson localization} applied within $B$, and recalling that constants $\Const$ are allowed to depend on $\xi$.

We need to get this estimate also with $\mathrm{L}$ instead of $\mathrm{R}$.  From
$$
(\widetilde{H}_\mathrm{L}-H_\mathrm{L} ) = - ( \widetilde{H}_\mathrm{R}- H_\mathrm{R})
$$
we get
$$
\gamma^{(\mathrm{L})}_{zw}=  -\gamma^{(\mathrm{R})}_{zw}
$$
which settles the desired estimate also for $\gamma^{(\mathrm{L})}_{zw}$ with $z=x-\ell$.  For the case of the other boundary, i.e.\  $z=x+\ell$, we run an analogous argument, but starting now with $\gamma^{(\mathrm{L})}$. Finally, the argument for  $\alpha$ instead of $\gamma$ proceeds analogously. 
\end{proof}

\subsubsection{Approximate solution of the Poisson equation $\mathcal{L} u_x =j_x$}
Next, we put the new left-right splitting to use to cast the local current $j_x$ a time-derivative $\caL u_x$, up to a small correction.

\begin{lemma}[localization]\label{lem: localization}
For $x \in G_0(\ell)$, there exist functions $u_x$ and $f_x$ so that 
\begin{equation}  \label{eq: poisson}
	\mathcal L u_x 
	\; = \;
	j_x + f_x,\qquad   \langle u_x \rangle_\beta=0,\qquad   \langle f_x \rangle_\beta=0  
\end{equation}
with the following properties:
\begin{enumerate}
\item
The function $u_x$ has support in $B$, depends only on $\omega_x, x \in B$, and satisfies
$$ \langle u_x^2 \rangle_\beta \leq C\ell^2. $$ 
\item    There is a random variable $r_x >0$ depending only on $\omega_y$ with $ |x-y| \leq \ell+1$ so that,
\begin{equation}
\label{eq: bound on fx}
   \econd(r_x)  \leq  \Const \ell  e^{-\ell/\xi} 
\end{equation}
and, for any $x,x' \in G_0(\ell)$,  such that $|x-x'| > 2\ell+2$
$$
\langle f^2_x  \rangle_\beta   \leq  \Const r^2_x,\qquad   |\langle f_x f_{x'} \rangle_\beta |  \leq  \Const r_x r_{x'}  \ed^{- \tfrac{1}{\zeta}(|x-x'|-(2\ell+2))}
$$
\end{enumerate}
\end{lemma}

\begin{proof}
We take 
$$
u_x=H_L-\widetilde{H}_L -  \langle H_L-\widetilde{H}_L \rangle_\beta
$$
which satisfies all the requirements of item $1$, because of the bounds in Lemma \ref{lem: bounds on gamma}. 
By the definition of $j_x$, we have that $\caL_B(H_L)+\caL_{B^c}(H_L)=\caL_B(H_L)=j_x$ and by construction $\caL_B(\widetilde{H}_L)=0$. Hence \eqref{eq: poisson} is satisfied by
$$
f_x := (\caL-\caL_B-\caL_{B^c}) u_x = \{ H_{\partial B}    ,u_x\},
$$
with 
$$ H_{\partial B}=    \tfrac12 g_0 ((q_{x+\ell+1} - q_{x+\ell})^2+ (q_{x-\ell} - q_{x-\ell-1})^2 ). $$ 
The zero-mean property $\langle f_x \rangle_\beta=0$ follows because $j_x$ and $\caL w$ (for any function $w$) have zero mean.  Using the explicit form of $u_x$, we get
$$
f_x =   2 g_0   \sum_{w \in B} p_w  \left( \gamma^{(\mathrm{L})}_{x+\ell,w}  ( q_{x+\ell}-q_{x+\ell+1})   +
 \gamma^{(\mathrm{L})}_{x-\ell,w}   (q_{x-\ell} - q_{x-\ell-1})
\right)
$$ 
 The bound on $f_x$ follows now by item $2$ of Lemma \ref{lem: bounds on gamma}, with $r_x:= \sum_{z \in \{-\ell,\ell\}} \sum_{w \in B} |\gamma^{(\mathrm{L})}_{z,w}| $, by using  $\langle q_z^2p_w^2 \rangle_\beta \leq \Const$ and the decay of correlations in Lemma \ref{lem: decay of correlations}.  
\end{proof}

\subsection{Griffiths regions}\label{sec: griffiths regions}

We now define a subset $G(\ell)$ of $G_0 (\ell)$ where the random variable $r_x$ defined in Lemma \ref{lem: localization} to bound the function $f_x$,  is exponentially small in $\ell$: 
\begin{equation}\label{eq: G set}
	G(\ell) 
	\; = \;
	\left\{ 
	x \in G_0(\ell) : r_x \le \ell^2 \ed^{- \ell/ \xi}
	\right\} 
\end{equation}
Given a realization of $(\omega_x)_{x\ge 1}$ and $(\tau_x)_{x\ge 1}$, 
we define the strictly increasing sequence $(g_i)_{i\ge 1}$ as well as a sequence $(d_i)_{i \ge 1}$ such that 
\begin{equation}\label{eq: g and d}
	\{ g_i \, : \, i \ge 1  \} 
	\; = \;  G(\ell)  \, , \qquad d_i = g_{i+1} - g_i \, , \quad i \ge 1 \, .  
\end{equation}
Finally, let 
$$
	n_L \; = \; \max \{i: g_i \le L \} \, .
$$
The next lemma expresses the fact that $d_i$ are not too large: 
\begin{lemma}\label{lem: di variables}
1. The variables $d_i,d_j$ are independent whenever $|i - j| \ge 2 \ell + 2$. \\
2. There exists $\ell_0 < + \infty$ such that for any $\ell \ge \ell_0$, any $i \ge 1$ and any $d \ge 1$,
$$
	P(d_i \ge d) \; \le \; \ed^{- d /d_0} 
	\qquad \text{with} 
	\qquad d_0 \; = \; \ed^{3 \ell \log \frac1{1-p}} \, .
$$
\end{lemma}
\begin{proof}
The first claim follows from the fact that $(\omega_x)_{x\ge 1}$ and $(\tau_x)_{x\ge 1}$ are independent sequences of i.i.d.\@ random variables,
that the event $x \in G(\ell)$ depends only on the variables $\omega_y,\tau_y$ with $|x-y| \le \ell$,
and that $d_i \ge 1$ for all $i \ge 1$. 

For the second claim, note that $P(x \in G_0(\ell)) \geq (1-p)^{2\ell+1}$ and that 
$$P(r_x > \ell^2 e^{-\ell/\xi}) \leq  C/\ell, $$
 by the bound \eqref{eq: bound on fx} in Lemma \ref{lem: localization} and the Markov inequality. Hence, for $\ell$ large enough, 
$$
P(x \in G(\ell)) \; \ge \; \frac12 (1-p)^{2\ell+1}.
$$
Since the events $x \in G(\ell)$ and $y \in G(\ell)$ are independent for $|x - y| \ge 2 \ell + 2$, we obtain
$$
	P (d_i \ge d) 
	\; \le \; 
	\left(1 - \frac12 (1-p)^{2 \ell +1} \right)^{d/(2 \ell +2)}
	\; \le \; 
	\exp{\big( -\frac{d (1-p)^{2 \ell +1}}{4(\ell +1)}\big)}
$$
which yields the claim for $\ell_0$ large enough.
\end{proof}

\begin{proof}[Proof of Theorem~\ref{th: disordered chain}]
Let us fix a realization of $(\omega_x)_{x\ge 1}$ and $(\tau_x)_{x\ge1}$. 
Let us fix some $t>0$ and some length $\ell = \ell (t)\in \N$ with a dependence on time that will be specified later on.
For simplicity, let us write $G$ for $G(\ell)$. 
We seek an upper bound on $C(t)$ defined in \eqref{eq: conductivity}.
Given $1 \le x < L$, let $x_{G} \in G$ be the closest point to $x$. 
For $x \notin G$, energy conservation implies that 
\begin{equation}\label{eq: current conserved energy}
	j_x 
	\; = \;
	j_{x_{G}}
	\; \pm \; \mathcal L \left( \sum_{y\in \{x,x_{G}\}} h_y \right)
\end{equation}
where 
$$
	\{ x, x_G \} \; = \; \{ x+1, \dots, x_G \} \quad \text{if} \quad x < x_G, 
	\qquad 
	\{ x, x_G \} \; = \; \{ x_{G+1}, \dots , x \} \quad \text{if} \quad x > x_G
$$
and $\pm = +$ if $x < x_G$, and $\pm = -$ if $x > x_G$. 
By Lemma \ref{lem: localization}, we write $j_{x_{G}} = \mathcal L u_{x_{G}} - f_{x_{G}}$, and we obtain
\begin{equation}
	\sum_{x=1}^{L-1} j_x
	\; = \; 
	-\sum_{x=1}^{L-1} f_{x_{G}}
	\; + \; \mathcal L \sum_{x=1}^{L-1} \Big( u_{x_{G}} \pm \sum_{y\in\{x,x_G\}} h_y \Big) \, .
\end{equation}
By Cauchy-Schwarz, we estimate
\begin{align} \label{eq: decomposition current}
	\frac{1}{2L} \left\langle\left( \int_0^t \dd s \sum_{x=1}^{L-1} \, j_x(s) \right)^2 \right\rangle_\beta  &\leq   	\frac{1}{L}\left\langle\left( \int_0^t \dd s \sum_{x=1}^{L-1} \, f_{x_G}(s) \right)^2  \right\rangle_\beta  &+&  \frac{1}{L}\left\langle\left( \sum_{x=1}^{L-1}  \int_0^t \dd s \sum_{x=1}^{L-1} \, \caL v_x(s) \right)^2  \right\rangle_\beta   \nonumber \\[3mm]
	&  =:  \qquad \qquad I_1(L,t) &+&  \qquad \qquad  I_2(L,t) 
\end{align}
where we abbreviated \begin{equation}\label{def: vx}
v_x=u_{x_{G}} \pm \sum_{y\in\{x,x_G\}} h_y. 
\end{equation}
To prove Theorem \ref{th: disordered chain}, we establish almost sure bounds on $\lim_L I_j(L,t)$, for $j=1,2$. \\
We start with $I_1(t)$. By Cauchy-Schwarz and stationarity, we have for any $w$,
$$
	\left\langle\left( \int_0^t \dd s \, w(s) \right)^2 \right\rangle_\beta
	\; \le \; 
	t^2 \langle w^2 \rangle_\beta.
$$
Taking $w= \sum_{x=1}^{L-1} \, f_{x_G} $, this yields here 
\begin{equation*}
	I_1 (L,t)
	\; \le \; 
	\frac{t^2}{L} \left\langle \left( \sum_{x=1}^L f_{x_G} \right)^2 \right\rangle_\beta 
	\; = \; 
	\frac{t^2}{L} \sum_{x,y \in \N_L} \langle f_{x_G} f_{y_G} \rangle_\beta 
	%\label{eq: bound part 1}
\end{equation*}
Using the variables $g_i, d_i$ defined in \eqref{eq: g and d}, we get then

$$
 \sum_{x,y \in \N_L} \langle f_{x_G} f_{y_G} \rangle_\beta  \leq   \sum_{1\leq i,j \leq n_L}   (d_i+d_{i+1})   (d_j+d_{j+1}) | \langle f_{g_i} f_{g_j} \rangle_\beta|
$$
By the definition  \eqref{eq: G set} of the set $G$, the bounds of Lemma \ref{lem: localization} and independence of the variables $r_{x}, r_{y}$ for $|x-y| \geq 2\ell+2$, we get 
$$
 \sum_{x,y \in \N_L} \langle f_{x_G} f_{y_G} \rangle_\beta
	\; \le \;
	\Const \ell^5 \ed^{-2\ell/\xi}  \sum_{1 \le i,j \le n_L}   \ed^{-\tfrac1\zeta |i-j|}  d_i d_j  
$$
Estimating $d_id_j \leq \tfrac12(d_i^2+d_j^2)$, we hence get
$$
	I_1 (L,t)
	\; \le \; 
	\Const t^2 \ell^5 \ed^{-2\ell/ \xi}  \frac1L \sum_{ i = 1}^{n_L} d_i^2 
$$
 By Lemma~\ref{lem: di variables}, the variables $d^2_i$ have finite moments and they have finite-range correlations. This suffices to establish a strong law of large numbers for the sequence $d_i^2$. Moreover, $E(d_i^2)\leq \Const d_0^2$ and hence we obtain the almost sure bound: 
$$
	I_1 (t)  \; := \; \overline{\lim}_{L \to \infty} I_1 (L,t) \; \leq \; \Const t^2   \ell^5   \ed^{b_1\ell},   \qquad  b_1=- \frac{2}{ \xi} + 6 \log \frac{1}{1-p} \, .
$$
Next, we consider the limit $t\to \infty $. We set, for some $a>0$ to be determined later,   
\begin{equation}\label{eq: l of t}
	\ell (t) \; = \;  a \log t. 
\end{equation}
and we obtain
$$
 I_1 (t) =\mathcal{O}((\log t)^5t^{2+ab_1}),\qquad t \to \infty
$$
\noindent We now move to $I_2(L,t)$, as defined in \eqref{eq: decomposition current}.
For any $v$ in the domain of $\caL$,
\begin{equation}\label{eq: trick no time}
	\left\langle \left( \int_0^t \dd s \, \mathcal L v(s) \right)^2 \right\rangle_\beta 
	\; = \; 
	\langle (\tilde v(t) - \tilde v(0))^2 \rangle_ \beta 
	\; \le \;
	2 \langle \tilde v^2 \rangle_\beta 
\end{equation}
with $\tilde v = v - \langle v \rangle_\beta$. 
We use this with $v=\sum_x v_x$ and $v_x$ as defined in \eqref{def: vx}. This yields
\begin{align*}%\label{eq: I2 term}
	I_2 (L,t)
	\; &\le \; 
	\frac{2}{L} 
	\left\langle \left(
	\sum_{x=1}^{L-1} \Big( u_{x_{G}} \pm \sum_{y\in\{x,x_G\}} \tilde h_y \Big)
	\right)^2 \right\rangle_\beta \, \\ 
	\; &\le \; 
	\frac4L \sum_{x,y \in \N_L} |\langle u_{x_G} u_{y_G} \rangle_\beta|
	\; + \; 
	\frac4L \sum_{x,y \in \N_L}  \sum_{\substack{x' \in \{ x,x_G \} , \\ y' \in \{ y,y_G\} }} |\langle \tilde h_{x'} \tilde h_{y'} \rangle_\beta| \, .
\end{align*}
Thanks to the decay of correlations in Lemma~\ref{lem: decay of correlations} and the properties of $u_{x_G}$ stated in Lemma~\ref{lem: localization}, we have
\begin{equation*}
	|\langle u_{x_G} u_{y_G} \rangle_\beta|
	\; \le \;
	\Const \ell^4 \ed^{- |x_G - y_G|/\zeta}\, ,
	\qquad
	|\langle \tilde h_{x'} \tilde h_{y'} \rangle_\beta|
	\; \le \; 
	\Const \ed^{- |x' - y'|/\zeta} \, .
\end{equation*}
Just as for $I_1(L,t)$, we estimate the number of $x$ such that $x_G=g_i$ by $d_i+d_{i+1}$ and we obtain
$$
	I_2 (L ,t)
	\; \le \; 
	\frac{\Const \ell^4}{ L} \sum_{1 \le i,j\le n_L} d_i d_j \ed^{- |i-j|/\zeta}  
	\; + \; 
	\frac\Const L \sum_{1 \le i,j\le n_L} d_i^2 d_j^2 \ed^{- |i-j|/\zeta} 
	\; \le \; 
	\frac{\Const \ell^4}{ L}  \sum_{i = 1}^{n_L} d_i^4
$$
where we used $d^m_id^m_j \leq d_i^m+d_j^m$ and $d_i \le d_i^2$ since $d_i \ge 1$.  
In the limit $L \to \infty$, we can again invoke the strong law of large numbers and we obain the almost sure bound
$$
	I_2(t) \; := \; \overline{\lim_{L \to \infty}} I_2(L,t) \; \le \; \Const  \ell^4 \ed^{b_2 \ell},\qquad 	b_2 \; = \; 12 \log\frac1{1-p} 
$$
We see hence that, for $\ell = \ell(t)$ given by \eqref{eq: l of t},   $I_2 (t) = O((\log t)^4t^{ab_2})$.\\
Adding the contributions of $I_1,I_2$, we conclude that
$$
C(t) \leq C (\log t)^5 \left( t^{2+ab_1}+ t^{ab_2}\right)
$$
The optimal $a$ is found by equating the two powers: $2+ab_1=ab_2$,
yielding
$$
a=  \frac{1}{v+1/\xi},\qquad  v=  3 \log (\frac1{1-p}).
$$
This yields the claims of the theorem.

\end{proof}

\section{Proof of Theorem~\ref{th: disordered fermion chain}}\label{sec: proof disordered fermion}

This proof is entirely analogous to that of Theorem \ref{th: disordered chain} but we still repeat the initial steps, because there are some superficial differences. 

We let constants $\Const$ depend on $J,g,\mu$ and the distribution of $\omega_x$.

\subsection{Observables and Gibbs state}\label{sec: preliminaries quantum}
Any observable $a$ can be expanded in a unique way as a linear combination of normally ordered monomials 
$$c^{\dagger}_{Y} = c^{\dagger}_{y_1}\ldots c^{\dagger}_{y_{p}},\qquad    c_{Y'}=  c_{y'_{p'}} \ldots c_{y'_1},\qquad p,p'\geq 0.$$
Here $Y,Y'$ are shorthand for finite tuples of $y$-coordinates. $p=0$ means that there are no annihilation operators, and analogously for $p'$, and hence $p=p'=0$ is the identity. 
Hence we can write
$$
a=\sum_{Y,Y'} a(Y,Y')   c^{\dagger}_{Y'} c_Y
$$
and we define the support of $a$ as 
$\mathrm{supp}(a)=Y \cup Y'$. 
In general,  let $H_X$ be the Hamiltonian restricted to a finite set $X \subset \N$, i.e.\ retaining only terms in \eqref{eq: fermionic chain} whose support is in $X$. Similarly,  we define
the restriction $N_X=\sum_{x \in X} n_x$
and we note that $[H_X,N_X]=0$. Finally, we also
define the restriction $\widetilde{\mathcal{H}}_X= \sum_{x,y \in X} \widetilde{\mathcal{H}}_{x,y}$ of $\widetilde{\mathcal H}$.
As an application of these definition, we check that, if, as will be assumed in the next section, the variables $\tau_y=0$ for $y \in B$ with $B$ a stretch of sites, then the restriction $H_B$ is  the second quantization of $\widetilde{\mathcal H}$;
\begin{equation}\label{eq: hb as free quantum}
H_B=\sum_{\{x,x+1\} \in B}    J (c^{\dagger}_{x+1 c_x }+   c^{\dagger}_x c_{x+1})    +\sum_{x \in B}\omega_x n_x =\sum_{x,x'}  c^\dagger_{x'} (\widetilde{\mathcal{H}}_B)_{x,x'}   c_x
\end{equation}

%where $\mathcal{H}_B$ is the Anderson Hamiltonian restricted to region $B$ and by $(\mathcal{H}_B)_{x,x'}  $ we simply denote its $(x,x')$ matrix element.
%This representation will be exploited in the next section.

Let us comment on the state $\langle \cdot \rangle_\mu$ defined in \eqref{def: quantum state}. This state is very easy to work with as it is the analogue of a product state on the lattice. Since the density matrix factorizes, $\ed^{-\mu N}=\prod_x \ed^{-\mu n_x}$, one easily establishes
\begin{lemma}\label{lem: decay correlations quantum}
\begin{enumerate}
\item Let a be an observable of the form $a=a(Y,Y')c^\dagger_Y c_{Y'}$, then 
$$\langle a   \rangle_\mu  \leq  ||a||  \leq |a(Y,Y')|$$
\item Whenever $\mathrm{supp}(a)\cap \mathrm{supp}(b)=\emptyset$, then 
$$\langle a  b  \rangle_\mu =0   
$$
\end{enumerate}
\end{lemma}
Here, the $||\cdot ||$ stands for the usual operator norm.

\subsection{Splitting of the quadratic Hamiltonian}

We now exploit the fact that the system is an Anderson insulator in the regions where the anharmonic potential is absent, i.e.\@ where $\tau_x=0$. The upshot is Lemma \ref{lem: bounds on gamma quantum}.  
Throughout this section, we fix a realization of $(\tau_x)_{x \ge 1}$. 
Given $\ell \in \N$, let
\begin{equation}\label{eq: G not set quantum}
	G_0(\ell) 
	\; = \;
	\{ 
	x \in \N \, : \,  \tau_y = 0 \text{ for all $y \in \N $ s.t.\  } |y-x|\le \ell 
	\} \, .
\end{equation} 
We also fix an element $x\in G_0(\ell)$ and we denote $B=\{x-\ell,\ldots, x+\ell\}$, such that $\tau_y=0$ for $y \in B$.
Until the end of this section, the expectation $\econd$ is assumed to be conditioned on $\tau_y=0, y \in B$ and we do not repeat this. 
\\

We first consider an obvious splitting of $N_B$ into a left ($\mathrm{L}$) and right ($\mathrm{R}$) part with respect to the midpoint $x$:
$$
	N_B \; = \; N_{\mathrm{L}} + N_{\mathrm{R}},\qquad N_\mathrm{L}:=N_{\{y\in B: y<x\}},\qquad  N_\mathrm{R}:=N_{\{y\in B: y\geq x\}}
$$
The observables $N_{\mathrm{L}},N_{\mathrm{R}}$ do not commute with $H_B$ in general, corresponding to the fact that particles can be transported. Our aim in this section is to find a modified left-right splitting $N_B=\widetilde{N}_\mathrm{L}+\widetilde{N}_\mathrm{R}$ which does satisfy $[H_B,\widetilde{N}_{\mathrm{L},\mathrm{R}}]=0$, reflecting the spatial localization of energy. 

We now recall the $H_B$ is quadratic in the $c_y,c^\dagger_y$-operators, which allows for an explicit analysis.  
As before, we let  $(\psi_k)_{k \in \mathcal{I}}$ be an orthonormal basis of eigenvectors of $\widetilde{\mathcal{H}}_B$, with $\mathcal{I}$ an index set, $|\mathcal{I}|=|B|$. 
We define the eigenmode operators
$$
c_k=\sum_k \psi_k(x) c_x,\qquad  c^{\dagger}_k=\sum_k {\psi_k(x)} c^{\dagger}_x,\qquad n_k= c^{\dagger}_k c_k
$$
Their two useful properties are
\begin{enumerate}
\item $[H_B,n_k]=0$
\item $N=\sum_{x}n_x= \sum_{k}n_k$
\end{enumerate}
The second property follows immediately from the Plancherel equality, and the first is an expression of the fact that $H_B$ is the second quantization of $\widetilde{\cal H}_B$, as exhibited in \eqref{eq: hb as free quantum}.
The splitting $N_B=\widetilde{N}_{\mathrm{L}} + \widetilde{N}_{\mathrm{R}} $ that we propose is defined by
\begin{equation}\label{eq: expression tilde quantum}
\widetilde{N}_{\mathrm{L}} =\sum_{y < x}\sum_{k}  |\psi_k(y)|^2n_k ,\qquad    \widetilde{N}_{\mathrm{R}} =\sum_{y \geq x}\sum_{k}  |\psi_k(y)|^2
 n_k 
 \end{equation}
 From the above properties, it is clear that this is indeed a splitting and that $\{H_B,\widetilde{N}_{\mathrm{L},\mathrm{\mathrm{R}}}\}=0 $, but it is not clear a-priori in which sense this splitting is similar to $N_B=N_{\mathrm{L}}+N_{\mathrm{R}} $ and we exhibit this now. 
Note first that both $N_{\mathrm{L},\mathrm{R}},\widetilde{N}_{\mathrm{L},\mathrm{R}} $ are linear combinations of $c^\dagger_{z} c_w$ with $z,w \in B$. 
Let us call 
\begin{equation} \label{eq: difference splittings quantum}
N_{\mathrm{L}}-\widetilde{N}_{\mathrm{L}} = \sum_{z,w \in B} \gamma^{(\mathrm{L})}_{z,w} c^\dagger_z  c_w , \qquad   N_{\mathrm{R}}-\widetilde{N}_{\mathrm{R}} = \sum_{z,w \in B} \gamma^{(\mathrm{R})}_{z,w} c^\dagger_z  c_w 
\end{equation}
The functions \eqref{eq: difference splittings quantum} are small in the sense that the coefficients $\gamma_{z,w}$ typically decay exponentially in $|z-x|+|w-x|$. We will not state this in full generality because the following lemma suffices for our purposes.
\begin{lemma}\label{lem: bounds on gamma quantum}

\begin{enumerate}
\item For any $z,w$, $|\gamma^{(\mathrm{L},\mathrm{R})}_{z,w}| <\Const$.
\item  If one of $z,w$ is at the boundary of $B$, (i.e.\ at distance $1$ of $B^c$), then 
$$
{E}(|\gamma^{(\mathrm{L},\mathrm{R}}_{z,w}|)   \leq \Const e^{-\ell/\xi}
$$
\end{enumerate}
\end{lemma}
\begin{proof}
From the above we derive an explicit expression for $\gamma^{\mathrm{R},\mathrm{R}}_{w,z} $ and this expression is identical to the one for the classical anharmonic case, i.e.\ equation \eqref{eq: gamma1} holds without any change.  We can therefore copy line per line the proof of Lemma  \ref{lem: bounds on gamma}. 

\end{proof}

Next, we put the new left-right splitting to use to cast the local current $j_x$ a time-derivative $j_x=\caL u_x$, up to a small correction.

\begin{lemma}[localization]\label{lem: localization quantum }
For $x \in G_0(\ell)$, there exist observables $u_x$ and $f_x$ so that 
\begin{equation}  \label{eq: poisson quantum}
	\mathcal L u_x 
	\; = \;
	j_x + f_x,\qquad   \langle u_x \rangle_\mu=0; ,\qquad   \langle f_x \rangle_\mu=0  
\end{equation}
with the following properties:
\begin{enumerate}
\item
The observable $u_x$ has support in $B$, depends only on $\omega_x, x \in B$, and satisfies
$$ \langle u_x^2 \rangle_\mu \leq \Const\ell^2. $$ 
\item    There is a random variable $r_x >0$ depending only on $\omega_y$ with $ |x-y| \leq \ell+1$ so that,
\begin{equation}
\label{eq: bound on fx quantum}
   {E}(r_x)  \leq  \Const \ell  e^{-\ell/\xi} 
\end{equation}
and, for any $x,x' \in G_0(\ell)$,  such that $|x-x'| > 2\ell+2$
$$
\langle f^2_x  \rangle_\mu   \leq  \Const r^2_x,\qquad   \langle f_x f_{x'} \rangle_\mu  =0
$$
\end{enumerate}
\end{lemma}

\begin{proof}
We take 
$$
u_x=N_L-\widetilde{N}_L -  \langle N_L-\widetilde{N}_L \rangle_\beta
$$
which satisfies all the requirements of $(1)$, because of the bounds in Lemma \ref{lem: bounds on gamma quantum}. 
By the definition of $j_x$, we have that $\caL_B(N_\mathrm{L})+\caL_{B^c}(N_\mathrm{L})=\caL_B(N_\mathrm{L})=j_x$ and by construction $\caL_B(\widetilde{N}_\mathrm{L})=0$. Hence \eqref{eq: poisson quantum} is satisfied by
$$
f_x := (\caL-\caL_B-\caL_{B^c}) u_x = \id [ H_{\partial B}    ,u_x],
$$
with 
$$ H_{\partial B}=   \sum_{y \in \{x-\ell-1,x+\ell \}}  J(c^\dagger_{y}c_{y+1}+hc)  + g \tau_y  n_y n_{y+1}
 $$ 
leading also to the correct support properties.
To check that $\langle f_x \rangle_\mu=0$, we use that $\langle\mathcal L w \rangle_\mu=0$ for any observable $w$ and that $\langle j_x \rangle_\mu=0$, see comment in Section \ref{sec: dynamics}.  Using the explicit form of $u_x$, we get
$$
f_x =   \sum_{z \in \{-\ell,\ell\}} \sum_{w \in B} \gamma_{z,w} [H_{\partial B}, c^\dagger_z  c_w    ]
$$ 
Putting $r_x:=\sum_{z \in \{-\ell,\ell\}} \sum_{w \in B} |\gamma_{z,w}| $, the claims follow by Lemma \ref{lem: bounds on gamma} $2)$,  the fact that $ || H_{\partial B} || \leq \Const$, and the product state property in Lemma \ref{lem: decay correlations quantum}.
\end{proof}

\subsection{Griffiths regions}  \label{sec: griffiths regions quantum}

The rest of the argument proceeds just as in the classical case, in Section \ref{sec: griffiths regions}. In particular, the definition of the set $G(\ell)$ and its main properties, i.e.\ Lemma \ref{lem: di variables}, is taken over without any change. Then, the reasoning in the rest of the proof for the classical anharmonic chain relies entirely on the representation  \eqref{eq: current conserved energy} for $x \notin G(\ell)$
$$
	j_x 
	\; = \;
	j_{x_{G}}
	\; \pm \; \mathcal L \left( \sum_{y\in \{x,x_{G}\}} h_y \right)
$$
(with the set $\{x,x_{G}\}$ defined below \eqref{eq: current conserved energy}). 
For the fermionic chain, we simply change this to 
$$
	j_x 
	\; = \;
	j_{x_{G}}
	\; \pm \; \mathcal L \left( \sum_{y\in \{x,x_{G}\}} n_y \right)
$$
The role of Lemma \ref{lem: decay of correlations} is played here by Lemma \ref{lem: decay correlations quantum}, and hence we can replace the correlation length $\zeta$ by $0$.

% acknowledgements
\bigskip
\noindent
\textbf{Acknowledgements.}
We are most grateful to David A.~Huse and J.~L.~Lebowitz who suggested the study of the disordered models introduced in this paper. 
The work of F.~H.\@ and S.O.\@ was partially supported by the grant ANR-15-CE40-0020-01 LSD of the French National Research Agency (ANR).
F.~H.\@ acknowledges also the support of the ANR under grant ANR-14-CE25-0011 EDNHS.
W.~D.~R.\@ acknowledges the support of the Flemish Research Fund FWO under grants
G098919N and G076216N, and the support of KULeuven University under internal grant C14/16/062.

% bibliography
\bibliographystyle{plain}
\bibliography{griffiths_phase.bib}

\begin{thebibliography}{10}

\bibitem{abdul_et_al_2017}
H.~Abdul-Rahman, B.~Nachtergaele, R.~Sims, and G.~Stolz.
\newblock {Localization properties of the disordered XY spin chain: A review of
  mathematical results with an eye toward many-body localization}.
\newblock {\em Annalen der Physik}, 529(7):1600280, 2017.

\bibitem{agarwal_et_al_2017}
K.~Agarwal, E.~Altman, E.~Demler, S.~Gopalakrishnan, D.~A. Huse, and M.~Knap.
\newblock {Rare-region effects and dynamics near the many-body localization
  transition}.
\newblock {\em Annalen der Physik}, 529(7):1600326, 2017.
\newblock 1600326.

\bibitem{agarwal_gopalakrishnan_knap_mueller_demler_2015}
K.~Agarwal, S.~Gopalakrishnan, M.~Knap, M.~M{\"u}ller, and E.~Demler.
\newblock {Anomalous diffusion and Griffiths effects near the many-body
  localization transition}.
\newblock {\em Physical Review Letters}, 114(16):160401, 2015.

\bibitem{alicki_fannes_2001}
R.~Alicki and M.~Fannes.
\newblock {\em {Quantum dynamical systems}}.
\newblock Oxford university press, 2001.

\bibitem{altman2015universal}
E.~Altman and R.~Vosk.
\newblock {Universal dynamics and renormalization in many-body-localized
  systems}.
\newblock {\em Annual Review of Condensed Matter Physics}, 6(1):383--409, 2015.

\bibitem{anderson_1958}
P.~W. Anderson.
\newblock Absence of diffusion in certain random lattices.
\newblock {\em Physical Review}, 109:1492--1505, 1958.

\bibitem{adp_2017}
J.~H. Bardarson, F.~Pollmann, U.~Schneider, and S.~Sondhi, editors.
\newblock {\em {Special Issue: Many-Body Localization}}, volume 529(7) of {\em
  Annalen der Physik}.
\newblock John Wiley \& Sons, 2017.

\bibitem{basko_2011}
D.~Basko.
\newblock {Weak chaos in the disordered nonlinear Schr{\"o}dinger chain:
  destruction of Anderson localization by Arnold diffusion}.
\newblock {\em Annals of Physics}, 326(7):1577--1655, 2011.

\bibitem{basko_aleiner_altshuler_2006}
D.~M. Basko, I.~L. Aleiner, and B.~L. Altshuler.
\newblock {Metal–insulator transition in a weakly interacting many-electron
  system with localized single-particle states}.
\newblock {\em Annals of Physics}, 321(5):1126--1205, 2006.

\bibitem{bernardin_huveneers_2013}
C.~Bernardin and F.~Huveneers.
\newblock {Small perturbation of a disordered harmonic chain by a noise and an
  anharmonic potential}.
\newblock {\em Probability Theory and Related Fields}, 157(1):301--331, 2013.

\bibitem{bodineau_helffer_2000}
T.~Bodineau and B.~Helffer.
\newblock {Correlations, Spectral Gap, and Log-Sobolev Inequalities for
  Unbounded Spin Systems}.
\newblock In R.~Weikard and G.~Weinstein, editors, {\em Studies in Advanced
  Mathematics, Differential Equations and Mathematical Physics}, volume~16,
  pages 51--66. AMS/IP, 2000.

\bibitem{casher_lebowitz_1971}
A.~Casher and J.~L. J.~L.~Lebowitz.
\newblock {Heat Flow in Regular and Disordered Harmonic Chains}.
\newblock {\em Journal of Mathematical Physics}, 12(8):1701--1711, 1971.

\bibitem{cycon_et_al_1987}
H.~L. Cycon, R.~G. Froese, W.~Kirsch, and B~Simon.
\newblock {\em {Schr\"odinger Operators with Applications to Quantum Mechanics
  and Global Geometry}}.
\newblock Springer-Verlag, 1987.

\bibitem{damanik_2011}
D.~Damanik.
\newblock {A Short Course on One-Dimensional Random Schr\"odinger Operators}.
\newblock {\em arXiv e-prints}, arXiv:1107.1094, 2011.

\bibitem{de_roeck_huveneers_2017}
W.~De~Roeck and F.~Huveneers.
\newblock {Stability and instability towards delocalization in many-body
  localization systems}.
\newblock {\em Physical Review B}, 95:155129, 2017.

\bibitem{de_roeck_huveneers_2019}
W.~De~Roeck and F.~Huveneers.
\newblock {Glassy dynamics in strongly anharmonic chains of oscillators}.
\newblock {\em arXiv e-prints}, page arXiv:1904.07742, 2019.

\bibitem{de2019sub}
Giuseppe De~Tomasi, Soumya Bera, Antonello Scardicchio, and Ivan~M Khaymovich.
\newblock Sub-diffusion in the anderson model on random regular graph.
\newblock {\em arXiv preprint arXiv:1908.11388}, 2019.

\bibitem{dhar_lebowitz_2008}
A.~Dhar and J.~L. Lebowitz.
\newblock {Effect of Phonon-Phonon Interactions on Localization}.
\newblock {\em Physical Review Letters}, 100:134301, 2008.

\bibitem{ducatez_2017}
R.~Ducatez.
\newblock {A forward--backward random process for the spectrum of 1D Anderson
  operators}.
\newblock {\em arXiv e-prints}, arXiv:1711.11302, 2017.

\bibitem{fleishman_anderson_1980}
L.~Fleishman and P.~W. Anderson.
\newblock {Interactions and the Anderson transition}.
\newblock {\em Physical Review B}, 21(6):2366, 1980.

\bibitem{froehlich_spencer_1983}
J.~Fr{\"o}hlich and T.~Spencer.
\newblock {Absence of diffusion in the Anderson tight binding model for large
  disorder or low energy}.
\newblock {\em Communications in Mathematical Physics}, 88(2):151--184, 1983.

\bibitem{gol1977pure}
I.~Ya. Gol'dshtein, S.~A. Molchanov, and L.~A. Pastur.
\newblock {A pure point spectrum of the stochastic one-dimensional
  Schr{\"o}dinger operator}.
\newblock {\em Functional Analysis and Its Applications}, 11(1):1--8, 1977.

\bibitem{gopalakrishnan_agarwal_demler_huse_knap_2016}
S.~Gopalakrishnan, K.~Agarwal, E.~Demler, D.~Huse, and M.~Knap.
\newblock {Griffiths effects and slow dynamics in nearly many-body localized
  systems}.
\newblock {\em Physical Review B}, 93(13):134206, 2016.

\bibitem{gornyi_mirlin_polyakov_2005}
I.~Gornyi, A.~Mirlin, and D.~Polyakov.
\newblock {Interacting electrons in disordered wires: Anderson localization and
  low-T transport}.
\newblock {\em Physical Review Letters}, 95(20):206603, 2005.

\bibitem{helffer_1998}
B.~Helffer.
\newblock {Remarks on Decay of Correlations and Witten Laplacians
  Brascamp--Lieb Inequalities and Semiclassical Limit}.
\newblock {\em Journal of Functional Analysis}, 155(2):571--586, 1998.

\bibitem{helffer_1999}
B.~Helffer.
\newblock {Remarks on decay of correlations and Witten Laplacians III.
  Application to logarithmic Sobolev inequalities}.
\newblock {\em Annales de l'Institut Henri Poincare, section B},
  35(4):483--508, 1999.

\bibitem{imbrie_jsp_2016}
J.~Imbrie.
\newblock {On many-body localization for quantum spin chains}.
\newblock {\em Journal of Statistical Physics}, 163(5):998--1048, 2016.

\bibitem{kozarzewski2018spin}
M.~Kozarzewski, P.~Prelov{\v{s}}ek, and M.~Mierzejewski.
\newblock {Spin subdiffusion in the disordered Hubbard chain}.
\newblock {\em Physical Review Letters}, 120(24):246602, 2018.

\bibitem{kunz_souillard_1980}
H.~Kunz and B.~Souillard.
\newblock {Sur le spectre des op\'erateurs aux diff\'erences finies
  al\'eatoires}.
\newblock {\em Communications in Mathematical Physics}, 78(2):201--246, 1980.

\bibitem{lepri_2016}
S.~Lepri, editor.
\newblock {\em {Thermal transport in low dimensions: from statistical physics
  to nanoscale heat transfer}}, volume 921 of {\em Lecture Notes in Physics}.
\newblock Springer, 2016.

\bibitem{lev2015absence}
Yevgeny~Bar Lev, Guy Cohen, and David~R Reichman.
\newblock Absence of diffusion in an interacting system of spinless fermions on
  a one-dimensional disordered lattice.
\newblock {\em Physical review letters}, 114(10):100601, 2015.

\bibitem{luitz_lev_2016}
D.~J. Luitz and Y.~Bar~Lev.
\newblock {Anomalous Thermalization in Ergodic Systems}.
\newblock {\em Physical Review Letters}, 117:170404, Oct 2016.

\bibitem{luitz_lev_2017}
D.~J. Luitz and Y.~Bar~Lev.
\newblock {The ergodic side of the many-body localization transition}.
\newblock {\em Annalen der Physik}, 529(7):1600350, 2017.

\bibitem{luitz_huveneers_de_roeck_2017}
D.~J. Luitz, F.~Huveneers, and W.~De~Roeck.
\newblock {How a Small Quantum Bath Can Thermalize Long Localized Chains}.
\newblock {\em Physical Review Letters}, 119:150602, 2017.

\bibitem{luitz_et_al_2016}
D.~J. Luitz, N.~Laflorencie, and F.~Alet.
\newblock {Extended slow dynamical regime close to the many-body localization
  transition}.
\newblock {\em Physical Review B}, 93(6):060201, 2016.

\bibitem{mendoza2019asymmetry}
Juan~Jose Mendoza-Arenas, M~{\v{Z}}nidari{\v{c}}, Vipin~Kerala Varma, John
  Goold, Stephen~R Clark, and Antonello Scardicchio.
\newblock Asymmetry in energy versus spin transport in certain interacting
  disordered systems.
\newblock {\em Physical Review B}, 99(9):094435, 2019.

\bibitem{mulanski_ahnert_pikovsky_shepelyansky_2009}
M.~Mulansky, K.~Ahnert, A.~Pikovsky, and D.~Shepelyansky.
\newblock Dynamical thermalization of disordered nonlinear lattices.
\newblock {\em Physical Review E}, 80:056212, 2009.

\bibitem{nachtergaele_reschke_2019}
B.~Nachtergaele and J.~Reschke.
\newblock {Slow propagation in some disordered quantum spin chains}.
\newblock {\em arXiv e-prints}, page arXiv:1906.10167, 2019.

\bibitem{oganesyan_huse_2007}
V.~Oganesyan and D.~A. Huse.
\newblock {Localization of interacting fermions at high temperature}.
\newblock {\em Physical Review B}, 75:155111, 2007.

\bibitem{oganesyan_pal_huse_2009}
V.~Oganesyan, A.~Pal, and D.~Huse.
\newblock {Energy transport in disordered classical spin chains}.
\newblock {\em Physical Review B}, 80(11):115104, 2009.

\bibitem{potter2015universal}
A.~C. Potter, R.~Vasseur, and S.~A. Parameswaran.
\newblock {Universal properties of many-body delocalization transitions}.
\newblock {\em Physical Review X}, 5(3):031033, 2015.

\bibitem{roy2018anomalous}
Sthitadhi Roy, Yevgeny~Bar Lev, and David~J Luitz.
\newblock Anomalous thermalization and transport in disordered interacting
  floquet systems.
\newblock {\em Physical Review B}, 98(6):060201, 2018.

\bibitem{rubin_greer_1971}
R.~J. Rubin and W.~L. Greer.
\newblock {Abnormal Lattice Thermal Conductivity of a One-Dimensional,
  Harmonic, Isotopically Disordered Crystal}.
\newblock {\em Journal of Mathematical Physics}, 12(8):1686--1701, 1971.

\bibitem{schulz2018energy}
Maximilian Schulz, Scott~Richard Taylor, Christopher~Andrew Hooley, and
  Antonello Scardicchio.
\newblock Energy transport in a disordered spin chain with broken {U(1)}
  symmetry: Diffusion, subdiffusion, and many-body localization.
\newblock {\em Physical Review B}, 98(18):180201, 2018.

\bibitem{serbyn_papic_abanin_2013}
M.~Serbyn, Z.~Papi{\'c}, and D.~Abanin.
\newblock {Local conservation laws and the structure of the many-body localized
  states}.
\newblock {\em Physical Review Letters}, 111(12):127201, 2013.

\bibitem{vznidarivc2016diffusive}
M.~{\v{Z}}nidari{\v{c}}, A.~Scardicchio, and V.~K. Varma.
\newblock {Diffusive and subdiffusive spin transport in the ergodic phase of a
  many-body localizable system}.
\newblock {\em Physical Review Letters}, 117(4):040601, 2016.

\end{thebibliography}

\end{document}